\documentclass[letterpaper, 10 pt, conference]{ieeeconf}  

\IEEEoverridecommandlockouts                              

\usepackage[utf8]{inputenc}
\usepackage{amsmath}

\usepackage{amsthm}
\usepackage{amsfonts}
\usepackage{amssymb}
\usepackage{graphicx}
\usepackage[outdir=./]{epstopdf}
\usepackage{xcolor}
\usepackage[pagebackref=false,breaklinks=true,colorlinks=true,bookmarks=false,allcolors=black,citecolor=black]{hyperref}
\usepackage{cite}
\usepackage{booktabs}
\usepackage{siunitx}
\usepackage[font=footnotesize]{caption}
\usepackage{subcaption}
\usepackage{tikz}

\newtheorem{proposition}{Proposition}
\newtheorem{problem}{Problem}

\newtheorem{remark}{Remark}
\newtheorem{theorem}{Theorem}

\newtheorem{lemma}{Lemma}
\newtheorem{assumption}{Assumption}

\newcommand{\etal}{\emph{et al.}}

\newcommand\copyrighttext{%
  \footnotesize © 2025 IEEE.  Personal use of this material is permitted. Permission from IEEE must be obtained for all other uses, in any current or future media, including reprinting/republishing this material for advertising or promotional purposes, creating new collective works, for resale or redistribution to servers or lists, or reuse of any copyrighted component of this work in other works.}
\newcommand\copyrightnotice{%
\begin{tikzpicture}[overlay, remember picture]
\node[anchor=south,yshift=10pt] at (current page.south) {\fbox{\parbox{\dimexpr\textwidth-\fboxsep-\fboxrule\relax}{\copyrighttext}}};
\end{tikzpicture}%
}

\title{\LARGE \bf
On dual-rate consensus under transmission delays
}

\author{David Umsonst$^{\dagger}$ and Mina Ferizbegovic$^{\dagger}$ 
\thanks{\hspace{-1.05em}\footnotesize \textsuperscript{$\dagger$} Ericsson Research, Stockholm, Sweden.\newline
        {\tt\scriptsize\{david.umsonst, mina.ferizbegovic\}@ericsson.com}}%
}

\begin{document}

\maketitle
\thispagestyle{empty}
\pagestyle{empty}

\begin{abstract}
In this paper, we investigate the problem of dual-rate consensus under transmission delays, where the control updates happen at a faster rate than the measurements being received. 
We assume that the measurements are delayed by a fixed delay and show that for all delays and rates, the system reaches a consensus if and only if the communication graph of the agents is connected and the control gain is chosen in a specific interval.
Based on these results we dive deeper into the convergence properties and investigate how the convergence changes when we change the rate for sending measurements. 
We observe that in certain cases there exists a sweet spot for choosing the sampling rate of the measurements, which can improve the convergence to the consensus point. 
We then formulate an optimization problem to find a sampling rate to improve the convergence speed and provide a necessary and sufficient condition for the existence of a finite optimizer of this problem. 
Our results are verified with numerical simulations.
\end{abstract}
\copyrightnotice
\section{Introduction}
Consensus has been extensively researched in the past decades and several results on the convergence of consensus algorithms have been proved \cite{ConsensusSurvey}.
Both event-triggered and time-triggered control methods have been employed to solve the consensus problem. In event-triggered control communication only happens if the predefined triggering condition is met, whereas in time-triggered control communication happens periodically. The differences of event-triggered and time-triggered control have been widely studied in \cite{astrom2002comparison, antunes2019consistent, mazenc2021event}.

When performing multi-agent control over a wireless network, such as 5G, one needs to take several network aspects into account. For example, event-triggered consensus strategies, such as \cite{seyboth2013event, nowzari2019event}, can be used to reduce the bandwidth requirements, although recent work has shown that time-triggered strategies could have a better performance\cite{TimeVsEventTriggeredConsensus}.
Network delays can also influence the performance of the consensus controller. Therefore, Atay~\cite{Atay} has investigated a consensus problem under transmission delays, where the agents have access to their own state but the state information of the other agents is delayed. Fan \etal~\cite{PRConsensusControl} look into the problem when the state information for the controller is delayed and propose a controller that introduces an additional delay to improve the convergence speed to the consensus value.
Ballotta and Gupta~\cite{ConsensusSparseControl} investigate whether or not sparse controllers can lead to a faster convergence of time-delayed consensus problem, where the time delay depends on the number of hops between the agents.

In the context of consensus, multi-rate control algorithms have not been studied extensively, although they are able to improve performance \cite{MultiRateControl} and can also reduce the required bandwidth by sending measurements with different rates.
An example application is the consensus of robots equipped with sensors that have different rates such as cameras and IMUs.

Thus, in this paper, we will look at dual-rate consensus, where the sampling rate of the sensors is lower than the sampling rate of the controller, where there is also a delay when sending the sensor measurements between the agents.
To the best of our knowledge, dual-rate consensus has not been intensively investigated compared to both single-rate time-triggered and event-triggered methods.
Furusaka \etal \cite{DualRateConsensusQuantizedCommunication} have a similar setup, where the measurements are sent with a lower rate than the controller, but their work looks into the design of a quantized dual-rate consensus controller and does not consider transmission delays.
Li \emph{et al.}~\cite{MultiRateControlWithDelayForConsensus} consider the control design of a delayed multi-rate system with an application to consensus. 
They aim to design a state-feedback matrix to stabilize the system for a given sampling time. Compared to our work it does not investigate the influence of the sampling time on the convergence.
The work closest to ours comes from Tanaka \etal \cite{DualRateConsensusJapanesePaper}, which investigates the stability of a dual-rate consensus problem. However, the results for stability are only sufficient conditions and neither delay nor the choice of a sampling rate are considered.

The contributions of our paper are threefold. First, we give a necessary and sufficient condition for the investigated dual-rate system to converge to a consensus, which only depends on the properties of the communication graph. 
Second, we analyze how the modes of the closed-loop system, unrelated to the consensus value, decay for different sampling periods of the measurements. 
Third, we propose an optimization problem to determine a sampling period that improves the convergence speed to the consensus point and give a necessary and sufficient condition when a finite optimizer exists.

The remainder of the paper is organized as follows. After introducing the notation at the end of this section, Section~\ref{sec:ProblemFormulation} introduces the system considered and states the problems, we investigate in this paper. Section~\ref{sec:AnalysingTheStability} looks at conditions for when a consensus can be achieved, while Section~\ref{sec:ConvergenceSpeed} investigates how the sampling period of the measurements influences the speed of the convergence. 
A numerical evaluation of our theoretical results is presented in Section~\ref{sec:NumericalExamples} and the paper is concluded in Section~\ref{sec:Conclusion}.

\textit{Notation:} Let $\mathbb{R}_{\geq a}$ and $\mathbb{N}_{\geq a}$ denote the real and integer numbers larger than $a$. For $x\in\mathbb{R}$, $\lceil x \rceil$, $\lfloor x \rfloor$, and $|x|$ are the smallest integer larger than or equal to, the largest integer smaller than or equal to, and the absolute value of $x$, respectively. 
Let $A$ and $B$ be sets then their union is denoted by $A\cup B$ and the number of elements in $A$ by $\#A$. Let $\mathtt{diag}(x_1,\ldots,x_n)$ the $n$-dimensional diagonal matrix with the $i$th diagonal element being $x_i\in\mathbb{R}$.
\section{Problem Formulation}
\label{sec:ProblemFormulation}
In this section, we will introduce the concept of dual-rate consensus and formulate the problem we aim to solve.

In consensus, the agents share their states with their neighbors and update their control input based on received states. 
In the single-rate case, the measurements are sent with the same rate as the control inputs are updated.
In the considered dual-rate case, the rate with which the measurements are sent is lower than the rate with which the control inputs are updated. 
An agent uses its own state and the last received states from its neighbors to determine a new control input.

We consider $n$ single-integrator agents discretized with a sampling period $T>0$
\begin{align}
    \label{eq:SingleIntegratorAgents}
    x_i((k+1)T) = x_i(kT) + Tu_i(kT),
\end{align}
where $x_i(kT)\in \mathbb{R}$, $u_i(kT)\in \mathbb{R}$, $k\in \mathbb{N}_{\geq 0}$, and  $i\in \lbrace 1,\ldots, n\rbrace$.
\begin{assumption}
\label{assum:GraphProperties}
The communication between the agents is represented as an \emph{undirected} and \emph{connected} graph ${\mathcal{G}=(\mathcal{V},\mathcal{E})}$, where the set of nodes $\mathcal{V}$ represents the agents and the set of edges $\mathcal{E}$ between the nodes represent which agents directly communicate with each other.
This graph is described by an adjacency matrix $A\in \mathbb{R}^{n\times n}$, where $a_{ij}=a_{ji}=1$ if the agents $i$ and $j\neq i$ are connected and otherwise $a_{ij}=0$. 
\end{assumption}
Similar to \cite{Atay}, the control input $u_i(kT)$ of the agents is given by
\begin{align}
    \label{eq:DelayDualRateConsensusLaw}
    u_i(kT) = \frac{\epsilon}{d_i}\sum_{j=1}^n a_{ij}\left(x_j(lT_s-T_d)-x_i(kT)\right),
\end{align}
where $d_i=\sum_{j=1}^n a_{ij}$ is the degree of vertex $i$ and $\epsilon$ is a controller gain.

Figure~\ref{fig:BlkDiagram} shows a block diagram of the proposed dual-rate setting.
While the state updates and controller inputs are sampled with sampling period $T$, indicated by time index $k$, the measurements of the other agents' states, $x_j$, are sampled with a different sampling period, i.e., $T_s$, and we use the index $l$ to indicate the different sampling rate. 
Having two different time indices to denote the fast and slow sampling rates is common in dual-rate consensus, see, for example, equation (7) in \cite{DualRateConsensusJapanesePaper}.
If $T_s= T$, we have a single-rate consensus problem as described above, where $k=l$, but if $T_s\neq T$ we have a dual-rate consensus problem, where $k\neq l$. 
In our work, we will establish a relation between $T$ and $T_s$ in Assumption~\ref{assum:SamplingPeriodRatio}.

Furthermore, delays might occur in the transmission of the state measurements, for example, if we assume  communication over a public network.
This delay in the state measurements is presented by $T_d\in \mathbb{R}_{\geq 0}$.
\begin{figure}
    \centering
    \includegraphics[width=0.5\textwidth]{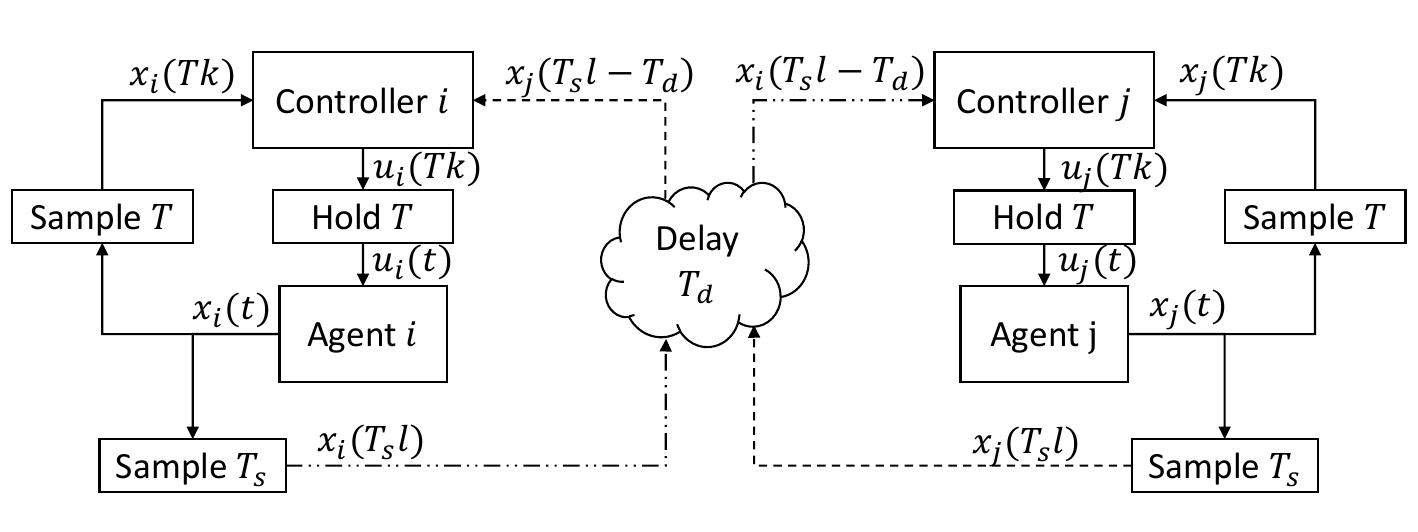}
    \caption{An overview of the dual-rate consensus setup with sampling period $T$ for the controller, sampling period $T_s$ for the transmitted measurements, and delay $T_d$.}
    \label{fig:BlkDiagram}
\end{figure}

We make the following assumptions on the rates and delay.
\begin{assumption}
    \label{assum:SamplingPeriodRatio}
    Let $T$ be the sampling period of the control input, then the measurements are sent with period $T_s=NT$, where $N\in \mathbb{N}_{\geq 1}$ is called the \emph{sampling period ratio}.
\end{assumption}
\begin{assumption}
    \label{assum:Delay}
    The delay is constant and a multiple of $T$, i.e., $T_d=hT$ and $h\in \mathbb{N}$.
\end{assumption}
In this work, we consider two problems. 
The first problem considers the feasibility of the consensus problem, i.e., if consensus can be achieved for certain $\epsilon$ and $N$ or not.
\begin{problem}
\label{prob:CanConsensusBeAchieved}
Given a fixed delay $h$, determine for which controller gains $\epsilon$ and sampling period ratios $N$ do the agents \eqref{eq:SingleIntegratorAgents} under the controller \eqref{eq:DelayDualRateConsensusLaw} achieve consensus.
\end{problem}
While the first problem addresses reaching a consensus, the second problem investigates how many measurements need to be sent to achieve an improved convergence to the consensus point.
\begin{problem}
\label{prob:WithWhatRateShouldWeCommunicate}
Given a fixed delay $h$ and a stabilizing controller gain $\epsilon$, determine $N$ such that the convergence to the consensus point is improved.
\end{problem}

\section{Convergence to the consensus value}
\label{sec:AnalysingTheStability}
To analyze the convergence to a consensus point, we begin by determining the closed-loop equations of the system with sampling period $T$ over a sampling period $T_s$, i.e., how the state of the $n$ agents evolves between $kNT$ and ${(k+1)NT}$.
For the sake of simplicity and without loss of generality we use $T=1$ in the remainder of the paper.

Based on Assumptions~\ref{assum:SamplingPeriodRatio} and \ref{assum:Delay}, we first determine the control input $u(kN+j)\in \mathbb{R}^n$ of all agents for $j\in[0,N)$, where the $i$th element of $u(kN+j)$ is the control input of agent $i$ as given in \eqref{eq:DelayDualRateConsensusLaw}.
To do so we need to determine the state measurements available at time step $kN$.
Hence, we find the smallest $\theta$ such that ${(k-\theta-1)N+h\leq kN}$ holds.
This results in $\theta=\lceil\frac{h}{N}\rceil-1$.
The next state measurement is received at $(k-\theta)N+h$, which in terms of the index $j$ is $j=h-\theta N\geq 0$.
Thus, based on \eqref{eq:DelayDualRateConsensusLaw} the control inputs to the $n$-dimensional system are given by
\begin{align}
u(Nk+j)=-\epsilon x(Nk+j)+\epsilon D^{-1}Ax\left(\left(k-\theta-1\right)N\right) 
\end{align}
for $0\leq j < \tau$ and 
\begin{align}
u(Nk+j)=-\epsilon x(Nk+j)+\epsilon D^{-1}Ax((k-\theta)N),
\end{align}
for $j \geq \tau$, where $\tau = h-\theta N$ and $D=\mathtt{diag}(d_1,\ldots, d_{n})$ is the diagonal matrix of vertex degrees.

This leads to the difference equation
\begin{equation}
\begin{aligned}
x((k+1)N) &= (1-\gamma) x(Nk) \\&+ \gamma \sum_{s=0}^1 f_s D^{-1}A x((k-\theta-s)N),
\end{aligned}
\end{equation}
where we introduced $\gamma=1-(1-\epsilon)^N$, 
\begin{align*}
f_0 = \frac{1-(1-\epsilon)^{N-\tau} }{1-(1-\epsilon)^N}\, \text{and}\ f_1=\frac{(1-\epsilon)^{N-\tau}(1-(1-\epsilon)^{\tau})}{1-(1-\epsilon)^N}.
\end{align*}
\begin{remark}
While $\gamma$, $f_0$, and $f_1$ are functions of $N$, we will omit this argument, for the sake of readability.
\end{remark}
\begin{lemma}
\label{lem:f1f0gamma}
For any $\epsilon\in(0,1)$, $N\geq 1$, and $h\geq 0$, it holds that $f_0+f_1=1$, $\gamma\in(0,1)$,  ${\lim_{N\rightarrow \infty}\gamma=1}$, ${\lim_{N\rightarrow \infty}f_0=1}$, and ${\lim_{N\rightarrow\infty}f_1=0}$. 
\end{lemma}
Based on Assumption~\ref{assum:SamplingPeriodRatio} we introduce a closed-loop system sampled at a larger period $T_s=NT$ given as follows
\begin{align}
\label{eq:SlowDynamicsWithLargeDelayAsInAtay}
x_s(l+1) = (1-\gamma) x_s(l) + \gamma \sum_{s=0}^1 f_s D^{-1}A x_s(l-\theta-s),
\end{align}
where we introduced the "slow" time step $l:=kN$ and ${l+1 := kN+N}$ related to sampling period $T_s$.
\begin{proposition}
\label{prop:ConsensusCondition}
Assume $\epsilon\in(0,1)$. The closed-loop dual-rate system reaches a consensus point if and only if the normalized Laplacian $L=I-D^{-1}A$ has a simple zero eigenvalue.
\end{proposition}
\begin{proof}
From Lemma~\ref{lem:f1f0gamma} we know that $f_0+f_1=1$, such that the system \eqref{eq:SlowDynamicsWithLargeDelayAsInAtay} has the same structure as the system given by equation (2.1) in \cite{Atay}, where we have only two delays here, i.e., $\theta$ and $\theta+1$.
Therefore, we can use Theorem 3.2 of \cite{Atay} to show that, under the assumption $\epsilon\in(0,1)$, the dual-rate consensus values is reached asymptotically for all $h\geq 0$ and $N\geq1$ if and only if the normalized Laplacian has a simple zero eigenvalue.
Furthermore, assuming $x(0)=x_s(0)$ we know that $x((k+\theta) N)=x_s(k+\theta)$ for all integers $k, \theta\in\mathbb{N}_{>0}$. Hence, both systems will converge to the same point if the  slow system is converging.
\end{proof}

Proposition~\ref{prop:ConsensusCondition} solves Problem~\ref{prob:CanConsensusBeAchieved} and, as in \cite{Atay}, we observe that the convergence to the consensus point does neither depend on the delay nor the sampling period ratio. It only depends on the structure of the communication graph, i.e., the graph needs to be connected for the agents to converge to a consensus point, as assumed in Assumption~\ref{assum:GraphProperties}. 
\section{Convergence properties}
\label{sec:ConvergenceSpeed}
In the previous section, we have shown that the dual-rate consensus converges for all delays $h\in\mathbb{N}_{\geq 0}$ and sampling period ratios $N\in\mathbb{N}_{\geq 1}$ given that $\epsilon\in(0,1)$.
In this section, we aim to address Problem~\ref{prob:WithWhatRateShouldWeCommunicate} and investigate how fast the consensus value is achieved based on the controller tuning $\epsilon$, the sampling period ratio $N$, and the delay $h$. 
For the derivations in this section, we assume $N\in\mathbb{R}_{\geq 1}$ and then discuss how the obtained results translate to $N\in\mathbb{N}_{\geq 1}$ in Section~\ref{sec:ImplicationsForIntergers}.

Following \cite{Atay}, we re-write the state $x_s(l)$ as
\begin{align}
    x_s(l)=\sum_{i=0}^n\alpha_i(l)\nu_i,
\end{align}
where $\alpha_i(l)=\mu_i^{\top}x_s(l)$, and $\nu_i$ and $\mu_i$ are the right and left eigenvector corresponding to the $i$th eigenvalue $\lambda_i$ of the normalized Laplacian, respectively.
Recall that the eigenvalues of the normalized Laplacian are real and $\lambda_i\in[0,2]$ holds for all $i$ \cite{SpectralGraphTheory}.
The dynamics of $\alpha_i(l)$ can be derived as
\begin{align}
    \alpha_i(l+1)=(1-\gamma)\alpha_i(l)+(1-\lambda_i)\gamma \sum_{s=0}^1 f_s\alpha_i(l-\theta-s).
\end{align}
Since this is a linear difference equation, we can apply the z-transform to it to obtain the characteristic equation
\begin{align}
    \label{eq:CharacteristicEquation}
    P_i(z)=z^{\theta+2}-(1-\gamma)z^{\theta+1}-(1-\lambda_i)\gamma f_0 z - (1-\lambda_i)\gamma f_1.
\end{align}
The absolute values of the largest roots of \eqref{eq:CharacteristicEquation} point towards the convergence speed of the system's mode $\lambda_i$.

Let us state the following result on the roots of \eqref{eq:CharacteristicEquation}.
\begin{lemma}
\label{lem:AbsoluteValuesOfRoots}
Assume $\epsilon \in (0,1)$, $h\geq 0$, and $N\geq 1$. If $\lambda_i=0$ then 1 is a simple root of \eqref{eq:CharacteristicEquation} and all other roots have an absolute value less than 1. If $\lambda_i\neq 0$, then the absolute value of all roots of \eqref{eq:CharacteristicEquation} are smaller than 1.
\end{lemma}
\begin{proof}
From Lemma~\ref{lem:f1f0gamma} we obtain ${\gamma\in(0,1)}$.
Thus, Lemma 3.1 of \cite{Atay} gives us the stated results.
\end{proof}

Let the modes be ordered as follows ${0=\lambda_0\leq \lambda_1\leq \cdots \leq \lambda_{n-1}\leq 2}$.
Lemma~\ref{lem:AbsoluteValuesOfRoots} shows us that the eigenvalue $\lambda_0$ has the largest absolute root at $1$ independent of the choice of $\epsilon$, $h$, and $N$, while the largest absolute value of the roots of the eigenvalues $\lambda_i\neq 0$ might change depending on the choice of $\epsilon$, $h$, and $N$.

This shows us that the root of $P_0(z)$ with value 1 relates to consensus value and the consensus value is reached faster the faster the dynamics related to the other roots of all $P_i(z)$ decay to zero.
Therefore, in the remainder of this paper, we will look at the largest absolute value of the roots related to the modes $\lambda_i\neq 0$ and the second largest absolute value of the root related to $\lambda_0=0$, because they are an indicator on how fast the agents converge to a consensus value. This is inspired by the dominant pole approximation method to analyse the convergence speed, which is similar to the second largest eigenvalues modulus approach used in other consensus papers, see, for example, \cite{ConsensusSparseControl}.

\subsection{Special case: $1\leq h\leq N$}
Next, we want to investigate some properties of the (second) largest absolute roots if $\lambda_i\neq 0$ ($\lambda_i=0$) for the case where the delay is smaller than the sampling period ratio and we are varying the sampling period ratio in $N\in[h,\infty)$. 

First, note that whenever $1\leq h\leq N$, the characteristic equation \eqref{eq:CharacteristicEquation} is a second-order polynomial, i.e.,
\begin{align}
    \label{eq:CharacteristicEquationWithNgreaterh}
     P_i(z)=z^{2}-\left((1-\gamma)+(1-\lambda_i)\gamma f_0\right) z - (1-\lambda_i)\gamma f_1,
\end{align}
since $\theta=0$, which implies $\tau=h$.

Next, let us analyze what happens to the coefficients of $P_i(z)$ when $N\rightarrow\infty$.
\begin{lemma}
\label{lem:ConvergenceOfAbsoluteValueOfRoots}
Given $\epsilon\in(0,1)$ and $1\leq h\leq N$, we obtain
\begin{align}
    \lim_{N\rightarrow\infty}P_i(z)=z(z-(1-\lambda_i)).
\end{align}
\end{lemma}
\begin{proof}
Using Lemma~\ref{lem:f1f0gamma} yields the stated result.
\end{proof}
This shows us that the largest absolute value of the roots converges to $|1-\lambda_i|$ and the other root converges to zero.

Subsequently, we will investigate properties of the (second) largest eigenvalues of the characteristic equation for different modes as $N$ increases.
\subsubsection{Considering $\lambda_0 = 0$}
For $\lambda_0=0$, we obtain the following result for the second largest absolute value.
\begin{proposition}
\label{prop:ConvergenceForLambda0}
Given $\epsilon\in(0,1)$ and $1\leq h\leq N$, the characteristic equation   \eqref{eq:CharacteristicEquationWithNgreaterh} for $\lambda_0=0$ is
\begin{equation}
    P_0(z)=(z-1)(z+\gamma f_1).
\end{equation}
This shows us that the second largest absolute value of the roots is given by $\bar{z}_0(N)=\gamma f_1$ and $\lim_{N\rightarrow \infty}\bar{z}_0(N)=0$.
\end{proposition}
\begin{proof}
From Lemma~\ref{lem:AbsoluteValuesOfRoots} we know that $P_0(z)$ has a root at $z=1$. This is used to obtain the factorization and, thus,  $\bar{z}_0(N)$. The limit for $\bar{z}_0(N)$ is obtained from Lemma~\ref{lem:f1f0gamma}.
\end{proof}
Hence, the second largest eigenvalue decays to zero, which means the larger $N$ is the faster the largest eigenvalue of the mode $\lambda_0$ dominates its dynamics.

\subsubsection{Considering $\lambda_i\in(0,1]$} 
For $\lambda_i\in(0,1]$ we will show that the largest absolute value of the roots of $P_i(z)$ is decreasing in the sampling period ratio $N$.

\begin{theorem}
\label{thm:DecreasingValue}
Assume $\lambda_i\in(0,1]$, $\epsilon\in(0,1)$, and $h\geq 1$ fixed. The largest absolute value of the root of the characteristic function $P_i(z)$ is decreasing in the sampling period ratio $N$ for $h\leq N$.
\end{theorem}
\begin{proof}
The proof can be found in Appendix~\ref{app:ProofOfDecreasingValueTheorem}.
\end{proof}

Theorem~\ref{thm:DecreasingValue} shows us that the larger we choose the sampling period ratio the faster the modes $\lambda_i\in(0,1]$ converge to zero. 
Thus, increasing the sampling period ratio has a positive effect on the convergence speed of the modes $\lambda_i\in(0,1]$ for $1\leq h\leq N$. 
Interestingly, this result tells us that the agents should never communicate to improve the convergence speed of these modes, since the optimal sampling period ratio for all $\lambda_i\in(0,1]$ is $N\rightarrow\infty$.

Finally, we want to show that the largest absolute value of the roots is decreasing in $\lambda_i$.
\begin{lemma}
\label{lem:LargestAbsoluteValueOfZeroAtSmallestEigenvalue}
Assume $\lambda_i,\lambda_j\in(0,1]$ and $\lambda_i\leq \lambda_j$. Then $\bar{z}_i(N)\geq\bar{z}_j(N)$ for all $N$ when $1\leq h\leq N$ and $\epsilon\in(0,1)$.
\end{lemma}
\begin{proof}
The proof can be found in Appendix~\ref{app:ProofRootIncreasingWithEigenvalue}.
\end{proof}

\subsubsection{Considering $\lambda_i\in(1,2]$}

In this section, we assume $\lambda_i\in(1,2]$ and investigate how the largest absolute value of the roots of $P_i(z)$ behaves when we increase the sampling period $N$ ratio for $h\geq 1$ and $\epsilon\in(0,1)$.

\begin{theorem}
\label{thm:MinimumExists}
Let $\lambda_i\in(1,2]$, $\epsilon\in (0,1)$ and $h\geq 1$.
The largest absolute value of the roots of $P_i(z)$ reaches its unique local minimum at the sampling period ratio $N=\frac{\log(g_1)}{\log(1-\epsilon)}$, where $g_1$ is the smallest root of
\begin{align}
    \label{eq:QuadraticEquationForOptimum}
    \frac{1}{4}\big(\left(1+b_ic\right)g-b_i\big)^2-gb_i(c-1)=0,
\end{align}
$b_i = |1-\lambda_i|$, and $c=(1-\epsilon)^{-h}$.
\end{theorem}
\begin{proof}
The proof can be found in Appendix~\ref{app:ProofOfTheoremAboutMinimumExistsForModes}.
\end{proof}
Theorem~\ref{thm:MinimumExists} shows us that for each mode $\lambda_i\in(1,2]$ there exists a unique sampling period ratio that minimizes the largest absolute value of the roots of \eqref{eq:CharacteristicEquationWithNgreaterh}. 
However, this sampling period ratio is not necessarily the same for all $\lambda_i\in(1,2]$, since the value of $g_1$ depends on $\lambda_i$. 
This can be seen in Figure~\ref{fig:OptimizationOtherGraph} in Section~\ref{sec:NumericalExamples}.

\subsection{Choosing the sampling period ratio $N$}

In the previous section, we analysed the convergence behavior of the modes $\lambda_i$. 
In this section, we will use these results to set up an optimization problem to find a sampling period ratio that improves the convergence speed according to the dominant pole approximation method, i.e., it minimizes the absolute value of the root with the largest absolute value.

Similar to before we define the (second) largest absolute value of the roots of a mode $\lambda_i$ as $\bar{z}_i(N)$.
Since the slowest modes approximately dominate the dynamics of the system, we want to find the optimal sampling period ratio by minimizing the objective,
\begin{align}
    \min_{N\geq h}\max\left( \bar{z}_0(N), \bar{z}_1(N),\cdots,\bar{z}_{n-2}(N),\bar{z}_{n-1}(N)\right),
\end{align}
where we have ordered the modes as follows, ${0=\lambda_0<\lambda_1\leq \lambda_2\leq \cdots \leq \lambda_{n-2} \leq \lambda_{n-1}}$ and recall that $n$ is the number of agents. 
This objective makes sure that the dominating absolute value of the roots is minimized, which in turn makes sure that the modes $\lambda_i$ will decay as fast as possible.
This is equivalent to minimizing the second largest eigenvalue modulus, which is a common metric to analyze the convergence speed of consensus problems \cite{ConsensusSparseControl}.

From Lemma~\ref{lem:LargestAbsoluteValueOfZeroAtSmallestEigenvalue} we know that $\bar{z}_1(N)\geq \bar{z}_i(N)$ for all $i>0$ such that $\lambda_i\in(0,1]$.
This simplifies the optimization objective to
\begin{align}
    \label{eq:OptimizationObjective}
    \min_{N\geq h}\max_{i\in \mathcal{I}_{>1}}\left( \bar{z}_0(N), \bar{z}_1(N),\bar{z}_i(N)\right),
\end{align}
where $\mathcal{I}_{>1}$ is the set of indices for which $\lambda_i\in(1,2]$.
Now we want to determine when \eqref{eq:OptimizationObjective} has a finite minimizer $N^*$.

\begin{theorem}
\label{thm:FiniteMinimizer}
Given $\epsilon\in(0,1)$ and $h\geq 1$, a finite minimizer $N^*$ of 
\begin{align}
    \min_{N\geq h}\max_{i\in \mathcal{I}_{>1}}\left( \bar{z}_0(N), \bar{z}_1(N),\bar{z}_i(N)\right),
\end{align}
exists if and only if $|1-\lambda_1| \leq |1-\lambda_{n-1}|$ where $\lambda_1 = \lambda_{n-1}$ needs to hold in the case of equality.
\end{theorem}

\begin{proof}
The proof can be found in Appendix~\ref{app:ProofOfFiniteMinimizer}
\end{proof}

Theorem~\ref{thm:FiniteMinimizer} shows us that for graphs that fulfill ${|1-\lambda_1|<|1-\lambda_{n-1}|}$ we can find a sampling period ratio $N^*$ that minimizes the largest absolute value of all roots.
Further, since we use the dominant pole approximation, the validity of this approximation depends on how fast the second root of $P_i(z)$ converges to zero. 
For large enough $N$, the roots of $P_i(z)$ are both real for all $i$ such that we expect that the dominant pole approximation should hold for large enough $N$. 
Similarly, the larger $\epsilon\in(0,1)$ is the larger $\gamma\in(0,1)$ is and the faster second root of $P_i(z)$ converges to zero for increasing $N$, which might indicate that $N^*$ gets closer to the true optimal rate, $N_{opt}$ for convergence as $\epsilon\rightarrow 1$.
In our numerical examples in Section~\ref{sec:NumericalExamples}, we investigate how close $N^*$ is to the sampling period ratio that minimizes the overall convergence speed of the system.

Finally, we would like to point out that our simulations indicate that $\bar{z}_{n-1}\geq \bar{z}_i$ for all $i\in \mathcal{I}_{>1}$, see, for example, Figure~\ref{fig:OptimizationOtherGraph}.
If this fact were true the objective \eqref{eq:OptimizationObjective} would simplify to
\begin{align}
    \max\left( \bar{z}_0(N), \bar{z}_1(N),\bar{z}_{n-1}(N)\right),
\end{align}
which means we would need to look at three instead of $\#\mathcal{I}_{>1}+2$ modes. Thus, the information needed from the graph is reduced significantly.
While we could not show that this result holds for the investigated dual-rate consensus problem, this has been shown for the single-rate case, see, for example, \cite{ConsensusSparseControl}. Therefore, we conjecture that this result holds for the dual-rate case as well.

\subsection{Implications of the results for integer-valued $N$}
\label{sec:ImplicationsForIntergers}
The derivations for the results in this section have so far assumed that $N\in \mathbb{R}_{\geq 1}$. Therefore, we want to discuss how these results translate to integer-valued sampling period ratios as assumed in Assumption~\ref{assum:SamplingPeriodRatio}.

Since Theorem~\ref{thm:DecreasingValue} states that the optimal $N$ for modes $\lambda_i\in(0,1]$ is infinitely large, this result also holds for integer-valued $N$.
In Theorem~\ref{thm:MinimumExists}, we show that for $1\leq h\leq N$ a unique minimum for each mode $\lambda_i\in(1,2]$. Since a unique minimum exists, the optimal sampling period ratio $N\in\mathbb{N}_{\geq 1}$ is either $\left\lceil\frac{\log(g_1)}{\log(1-\epsilon)}\right\rceil$ or $\left\lfloor\frac{\log(g_1)}{\log(1-\epsilon)}\right\rfloor$ depending which one leads to a lower value for the largest absolute value of the roots, where $g_1$ is the smallest root of \eqref{eq:QuadraticEquationForOptimum}.
Finally, Theorem~\ref{thm:FiniteMinimizer} states that under a certain condition the optimization problem has a finite real-valued optimizer. This implies that a finite integer-valued optimizer exists as well.
Hence, the results obtained in this section still hold under Assumption~\ref{assum:SamplingPeriodRatio} and Assumption~\ref{assum:Delay}.

\section{Numerical Examples}
\label{sec:NumericalExamples}
In this section, we will numerically evaluate the theoretical results obtained in the previous sections. We look at a graph with ${n=6}$ agents and the following initial condition and adjacency matrix 
\begin{align}
    \label{eq:InitialState}
     x(0)=\begin{bmatrix}5 \\ 6 \\ -3.5 \\ 0 \\ -2 \\ 3\end{bmatrix}\ \text{and}\ 
     A=\begin{bmatrix}
        0& 1& 0& 0& 1& 0\\
        1& 0& 1& 0& 1& 0\\
        0& 1& 0& 1& 0& 0\\
        0& 0& 1& 0& 1& 1\\
        1& 1& 0& 1& 0& 0\\
        0& 0& 0& 1& 0& 0
    \end{bmatrix}.
\end{align}

The controller is set to ${\epsilon=0.3}$ and we assume a delay of $h=10$.
For this graph the eigenvalues of the normalized Laplacian are given by ${\lambda_i\in \lbrace 0, 0.446, 0.871, 1.284, 1.521, 1.877\rbrace}$.
Based on the eigenvalues, Theorem~\ref{thm:FiniteMinimizer} shows us that a finite minimizer of \eqref{eq:OptimizationObjective}  exists.

The upper plot of Figure~\ref{fig:OptimizationOtherGraph} shows ${\bar{z}_i(N)}$ for the different modes $\lambda_i$ and $i>0$.
Note that for all modes with $\lambda_i\leq 1$, $\bar{z}_i(N)$ decreases with an increasing sampling period ratio $N$ while for all modes with $\lambda_i\in(1,2]$ a unique minimum for each $\bar{z}_i(N)$ exists. 
Both these observations verify our theoretical results of Theorem~\ref{thm:DecreasingValue} and Theorem~\ref{thm:MinimumExists}.
The lower plot of Figure~\ref{fig:OptimizationOtherGraph} shows the function $\max\left( \bar{z}_1(N),\ldots,\bar{z}_{n-1}(N)\right) $ and we see that the optimal sampling period ratio is given by $N^*=16$ for $N\geq h$. 
Furthermore, we observe that $N^*$ is the global optimum over all $N\in[1,50]$.
\begin{figure}
    \centering
    \begin{minipage}{0.45\textwidth}
        \centering
        \includegraphics[width=1\textwidth]{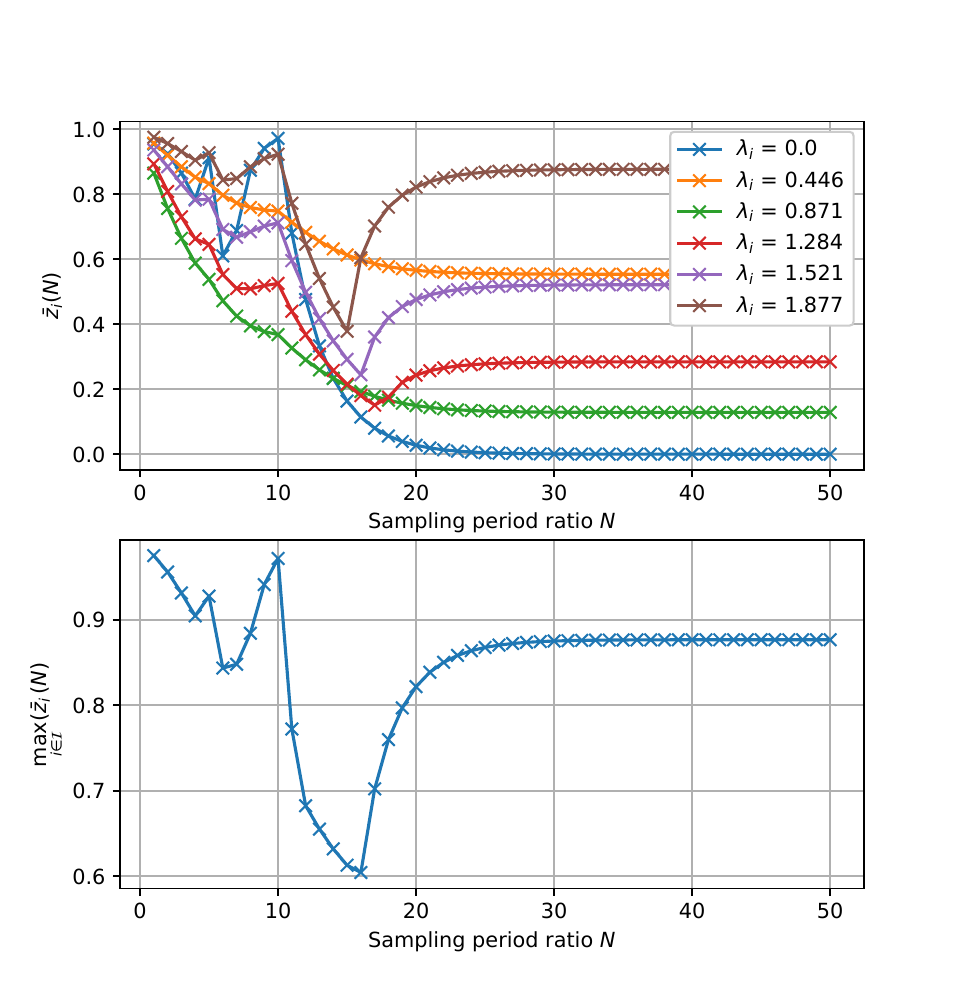} 
        \caption{The upper plot shows the largest absolute value of the roots $\bar{z}_i(N)$ corresponding to the modes $\lambda_i$ of the investigated graph for different $N$. The lower plot shows the objective \eqref{eq:OptimizationObjective} over $N$.}
        \label{fig:OptimizationOtherGraph}
    \end{minipage}\hfill
    \begin{minipage}{0.45\textwidth}
        \centering
        \includegraphics[width=\textwidth]{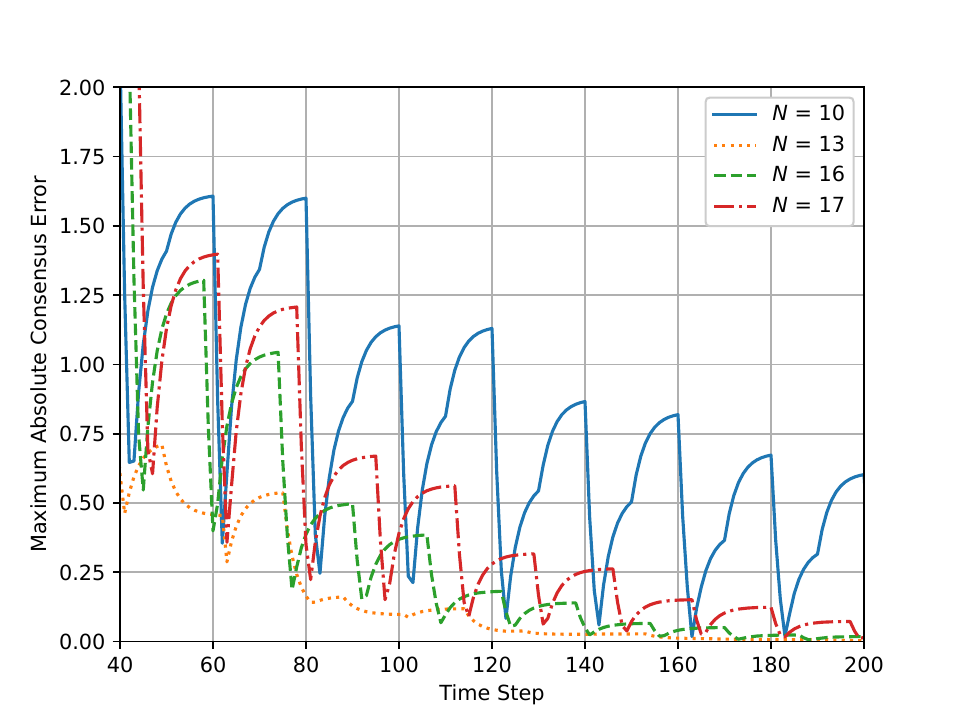} 
        \caption{This plot shows the maximum absolute error from the consensus point over time for different sampling period ratios ${N\in\lbrace 10, 13, 16, 17 \rbrace}$.}
        \label{fig:MaximumErrorTrajectoryOtherGraph}
    \end{minipage}
\end{figure}
Figure~\ref{fig:MaximumErrorTrajectoryOtherGraph} shows the trajectory for the maximum absolute error from the consensus point for different sampling period ratios.
The error converges faster for $N^*=16$ than for $N\in\lbrace10, 17\rbrace$, but it converges slower than for $N=13$.

The results above show that by solving the optimization problem \eqref{eq:OptimizationObjective} we obtain a sampling period ratio $N^*$ which improves the performance over setting $N=h$. However, it is not the optimal $N$, since there are values in $[h,N^*]$ for which the agents converged faster.
This is an expected result because we used the dominant root of the characteristic equation to investigate the convergence speed, which is an approximation of the convergence speed.

In the following, we look at the convergence speed of the agents for $\epsilon=0.1+0.1j$, ${j\in\lbrace0, \ldots, 8\rbrace}$, and a delay for $h=10$, and compare the convergence speed when using $N^*$ to the convergence speed of the optimal $N$, which will be determined empirically.
To determine the sampling period ratio that leads to the fastest speed of convergence we compute the sampling period ratio for which the solution of
\begin{align}
    \min_{k\in \mathbb{N}_{\geq 0}} k\ \mathrm{s.t.}\ \max_{i\in\lbrace 1,\cdots,n\rbrace} x_i(j)- \min_{i\in\lbrace 1,\cdots,n\rbrace} x_i(j)\leq \delta\ \mathrm{for}\ j\geq k
\end{align}
is the smallest, where we used $\delta=10^{-5}$.
Here, $N_{opt}$ and $N_{opt}^{N\geq h}$ denote the optimal sampling period ratio for all ${N\in \mathbb{N}_{\geq 1}}$ and $N\in \mathbb{N}_{\geq h}$, respectively.
The results are shown in Table~\ref{tab:OptimalSamplingPeriodRatios}.
As expected the larger $\epsilon$ is the closer the optimal solution $N^*$ of the optimization problem \eqref{eq:OptimizationObjective} is to $N_{opt}^{N\geq h}$. Furthermore, we also see that the optimal value $N_{opt}$ over all $N$ is the same as $N^*$ for larger $\epsilon$.
Furthermore, we see that $N_{opt}^{N\geq h}\leq N^*$, which might hint at $N^*$ being an upper bound on the true optimal value.
Interestingly, for $\epsilon\in\lbrace 0.1, 0.2 \rbrace$ it is optimal to send state measurements at each time step.
\begin{table}[]
    \centering
    \begin{tabular}{c|c|c|c|c|c|c|c|c|c}
         $\epsilon$ & $0.1$ & $0.2$ & $0.3$ & $0.4$ & $0.5$ & $0.6$ & $0.7$ & $0.8$ & $0.9$ \\
         \hline
         $N^*$ & 29 & 19 & 16 & 14 & 13 & 12 & 11 & 11 & 11\\
         \hline
         $N_{opt}^{h\leq N}$ & 10 & 13 & 14 & 14 & 12 & 12 & 11 & 11 & 11\\
         \hline
         $N_{opt}$ & 1 & 1 & 14 & 14 & 12 & 12 & 11 & 11 & 11\\
         \hline
         
    \end{tabular}
    \caption{In this table, we compare the finite solution $N^*$ of the optimization problem \eqref{eq:OptimizationObjective} to the true optimal sampling period ratio $N_{opt}$ and the optimal sampling period ratio $N_{opt}^{h\leq N}$ for $N\geq h$ for different $\epsilon$ and $h=10$.}
    \label{tab:OptimalSamplingPeriodRatios}
\end{table}

Next, we examine the behaviour of the dual-rate consensus problem for a graph where no finite optimizer exists according to Theorem~\ref{thm:FiniteMinimizer}. Here, we look at a graph with adjacency matrix
\begin{align}
     A=\begin{bmatrix}
        0 & 1 & 0 & 1 & 1 & 0\\
        1 & 0 & 1 & 0 & 1 & 1\\
        0 & 1 & 0 & 0 & 0 & 1\\
        1 & 0 & 0 & 0 & 1 & 0\\
        1 & 1 & 0 & 1 & 0 & 0\\
        0 & 1 & 1 & 0 & 0 & 0
    \end{bmatrix},
\end{align}
which leads to $|1-\lambda_1|=0.658>|1-\lambda_{n-1}|=0.621$. Hence, the condition given in Theorem~\ref{thm:FiniteMinimizer} is not fulfilled and the optimization problem \eqref{eq:OptimizationObjective} would tell us that we should never communicate. 

For the simulations, we use $\epsilon=0.1+0.1j$, $j\in\lbrace0, \ldots, 8\rbrace$, $h=10$, and the same initial state as in \eqref{eq:InitialState}.
We will determine the optimal period ratios empirically as done in Table~\ref{tab:OptimalSamplingPeriodRatios}.

Interestingly, for all investigated $\epsilon$ the optimal sampling period ratio shows that the agents should send the measurements at the same rate as the controller is updated, i.e. $N_{opt}=1$ 
This opens the question if single-rate consensus is always a better option than dual-rate consensus for graphs, for which $|1-\lambda_1|>|1-\lambda_{n-1}|$ holds.
\section{Conclusion}
\label{sec:Conclusion}
In this work, we investigated the dual-rate consensus problem for single-integrator agents, where the state measurements from neighbouring agents are sent with a lower rate than the rate the controller signals are updated with. In addition to the lower rate for the measurements the measurements are subject to a constant delay as they are sent between the agents.

We showed that as long as the graph representing the connections between agents is connected, a consensus point is reached regardless of the value of the delay and the rate with which the measurements are sent.
Inspired by the dominant pole approximation approach, we examined the convergence speed of the modes when the sampling period ratio is larger than the delay. We showed for certain modes not communicating leads to the fastest convergence, measured by the largest absolute value of the roots related to the mode, while for other modes there exists a unique sampling period ratio for each mode that improves the convergence.
With these results, we formulated an optimization problem to find a sampling period ratio to send the state measurements with and derived a necessary and sufficient condition for a finite minimizer of this problem to exists.
Our numerical examples verified the theoretical results and we also showed that the larger the controller gain the more likely it is that our dominant pole approximation approach holds.
In addition to that, for graphs that do not fulfill the condition of Theorem~\ref{thm:FiniteMinimizer} using a single-rate consensus approach seems to be the best choice according to our simulation results.

There are several avenues for future work. 
First,  our results only looked at the case where the sampling period ratio is larger than the delay. Therefore, it would be interesting to investigate the case where the sampling period is smaller than the delay.
Second, according to our simulations it seems like the optimization problem we pose to find the sampling period ratio leads to an upper bound on the optimal rate. 
Looking deeper into this would also benefit in limiting the search space for finding the optimal sampling period ratio.
Third, we investigated homogeneous single-integrator system such that the investigation of systems with heterogeneous agents and individual sampling period ratios could lead to compelling results on how to choose the sampling period ratios. Fourth, looking at time-varying delay and packet drops would also bring the work closer to real-world communication systems.
\appendix
\subsection{Proof of Theorem~\ref{thm:DecreasingValue}}
\label{app:ProofOfDecreasingValueTheorem}
We begin by simplifying the notation. We introduce ${b_i=1-\lambda_i}$, $c=(1-\epsilon)^{-h}$ and $g=(1-\epsilon)^N$, which leads to the roots for $P_i(z)$,
\begin{align}
    z_{1/2}(g)=p_1(g)\pm\sqrt{p_1(g)^2+p_2(g)},
\end{align}
where
\begin{align}
    p_1(g)=\frac{(1-b_ic)g+b_i}{2}\quad \mathrm{and}\quad p_2(g)=b_ig(c-1).
\end{align}
For the case $\lambda_i\in(0,1]$, the largest absolute value of the roots of $P_i(z)$ is given by
\begin{align}
    \bar{z}_i(g)=p_1(g)+\sqrt{p_1(g)^2+p_2(g)},
\end{align}
which abuses notation slightly.
With the simplified notation, we would like to investigate how $\bar{z}_i(g)$ changes with $g$. 

To do so, we look at the derivative of $\bar{z}_i(g)$, which is
\begin{align}
    \frac{d}{dg}\bar{z}_i(g)=\frac{d}{dg}p_1(g)+\frac{d}{dg}\sqrt{p_1(g)^2+p_2(g)}\\
    =\frac{1}{2\sqrt{p_1(g)^2+p_2(g)}}\left((1-b_ic)\bar{z}_i(g)+b_i(c-1)\right)
\end{align}
Next let us investigate when the derivative is greater than zero.
\begin{align}
    \frac{d}{dg}\bar{z}_i(g)&> 0\quad \Leftrightarrow\quad (1-b_ic)\bar{z}_i(g)> b_i(1-c)
\end{align}
Note now that $b_i(1-c)\leq 0$, since $c=(1-\epsilon)^{-h}>1$, $b_i=1-\lambda_i\geq 0$, and $\lambda_i\in(0,1]$. So if $1-b_ic\geq 0$, this condition is always fulfilled, because $\bar{z}_i(g)\geq 0$.
Hence, we will focus now on the case $1-b_ic<0$, which gives us
\begin{align}
    \label{eq:ConditionDecreasingValueOfRoots}
    \bar{z}_i(g)< \frac{b_i-b_ic}{1-b_ic}=1+\frac{b_i-1}{1-b_ic}.
\end{align}
Recall that $b_i-1< 0$ and by assumption $1-b_ic<0$. Therefore, the second summand on the right side of the last equality is positive. 
Since we know that $\bar{z}_i(g)<1$ for ${\lambda_i\in(0,1]}$ (see Lemma~\ref{lem:AbsoluteValuesOfRoots}), \eqref{eq:ConditionDecreasingValueOfRoots} holds. Hence, $\bar{z}_i(g)$ is increasing in $g$, which means it is decreasing in $N$ and this proves Theorem~\ref{thm:DecreasingValue}.

\subsection{Proof of Lemma~\ref{lem:LargestAbsoluteValueOfZeroAtSmallestEigenvalue}}
\label{app:ProofRootIncreasingWithEigenvalue}

Recall from the proof of Theorem~\ref{thm:DecreasingValue} that the largest absolute values are given by
\begin{align}
    \bar{z}_i(g)=p_1^i(g)+\sqrt{p_1^i(g)^2+p_2^i(g)},
\end{align}
where $i$ refers to the $i$th eigenvalue, $\lambda_i$.
Now if we can show that $p_1^i(g)\geq p_1^j(g)$ and $p_2^i(g)\geq p_2^j(g)$ for all $g$ then ${\bar{z}_i(g)\geq\bar{z}_j(g)}$ for all $g$ holds as well.
Let us start with $p_1^i(g)\geq p_1^j(g)$,
\begin{align}
    \frac{(1-b_ic)g+b_i}{2}&\geq \frac{(1-b_jc)g+b_j}{2},\\
    \Leftrightarrow (b_i-b_j)(1-cg)&\geq 0,\\ 
    \Leftrightarrow 1&\geq cg,
\end{align}
where we used that $b_i-b_j>0$ whenever $\lambda_i<\lambda_j$. Since $h\leq N$, we have $cg=(1-\epsilon)^{N-h}<1$, which shows us that the last inequality holds.
Next, we look at $p_2^i(g)\geq p_2^j(g)$,
\begin{align}
    b_ig(c-1) \geq b_jg(c-1) \Leftrightarrow b_i \geq b_j \Leftrightarrow \lambda_j \geq \lambda_i,
\end{align}
where we used that $c>1$.
This concludes the proof.

\subsection{Proof of Theorem~\ref{thm:MinimumExists}}
\label{app:ProofOfTheoremAboutMinimumExistsForModes}

Similar to the proof of Theorem~\ref{thm:DecreasingValue}, we begin simplifying the notation by introducing $g=(1-\epsilon)^N$, $b_i=|1-\lambda_i|$, and $c=(1-\epsilon)^{-h}$, which leads to
\begin{align}
    P_i(z)=z^2-2p_1(g)z+p_2(g),
\end{align}
where
\begin{align}
    p_1(g)=\frac{(1+b_ic)g-b_i}{2}\quad \mathrm{and}\quad p_2(g)=b_ig(c-1).
\end{align}
The roots of the polynomial are then given by
\begin{align}
    z_{1/2}(g)=p_1(g)\pm\sqrt{p_1(g)^2-p_2(g)}
\end{align}

Recall that $g\in(0,\frac{1}{c}]$, since $N\geq h$.
This shows us that $p_1(\frac{1}{c})=\frac{1}{2c}\geq 0$ and $\lim_{g\rightarrow 0 }p_1(g)=-\frac{b_i}{2}$. Hence, we observe that $p_1(g)$ has a sign change as $g$ decreases, which indicates that the root with the largest absolute value does not have a simple solution as in the case $\lambda_i\in(0,1]$, where we could simply use the positive sign in front of the square root for the largest absolute root.

Next, we will check $p_1(g)^2-p_2(g)$, which is a continuous function in $g$, to see if we could have complex-valued roots.
Let us define $g_0=\frac{b_i}{1+b_ic}$ such that $p_1(g_0)=0$.
Then we know that $p_1^2(g_0)-p_2(g_0)=-p_2(g_0)<0$.
Furthermore, $p_1(0)^2-p_2(0)=\frac{b^2}{4}$.
This show us that $p_1(g)^2-p_2(g)$ has one zero $g_1\in[0,g_0]$ and another zero $g_2 \in[g_0,\infty)$. This indicates that complex-valued roots are possible.
Since we are only interested in value of $g\in[0,\frac{1}{c}]$, let us now consider two different cases
\begin{itemize}
    \item \textbf{Case 1:} Let $b_i$ and $c$ be such that $g_2\geq \frac{1}{c}$ holds:
    
    In this case, we have $p_1(\frac{1}{c})^2-p_2(\frac{1}{c})<0$. Thus, the roots of $P_i(z)$ are complex-conjugated and their absolute value is given by
    \begin{align}
        \bar{z}_i(g)=\sqrt{p_2(g)}=\sqrt{b_ig(c-1)},
    \end{align}
    on the interval $[g_1,\frac{1}{c}]$, which is increasing in $g$. Hence, the minimum value of $\bar{z}_i(g)$ is located at $g_1$ for ${g\in[g_1, \frac{1}{c}]}$.
    
    For $g\in[0,g_1)$, we observe that $p_1(g)<0$ and ${p_1(g)^2-p_2(g)>0}$, which leads to 
    \begin{align}
        \bar{z}_i(g)=\max_i z_i(g) = -p_1(g)+\sqrt{p_1(g)^2-p_2(g)}.
    \end{align}
    To see if $\bar{z}_i(g)$ is increasing for $g\in[0,g_1)$ we look at its derivative. 
    The derivative of $\bar{z}_i(g)$ is given by
    \begin{align}
        \frac{d}{dg}\bar{z}_i(g) = -\frac{d}{dg}p_1(g)+\frac{d}{dg}\sqrt{p_1(g)^2-p_2(g)}\\
        =\frac{1}{2\sqrt{p_1(g)^2-p_2(g)}}\left(-(1+b_ic)\bar{z}_i(g)-b_i(c-1)\right).
    \end{align}
    Next, we want to determine under which conditions the derivative is positive.
    \begin{align}
        \frac{d}{dg}\bar{z}_i(g)\geq 0 \Leftrightarrow \bar{z}_i(g)\leq -\frac{b_i(c-1)}{1+b_ic}
    \end{align}
    Since $\bar{z}_i(g)\geq 0$ and the right-hand side of the last inequality is negative, we see that this inequality will never hold.
    Thus, the maximum absolute value is decreasing in $g$ for $g\in [0,g_1)$, which means it is increasing in $N$.

    Hence, in this case, the minimum value of $\bar{z}_i(g)$ on ${g\in[0,\frac{1}{c}]}$ is achieved at $g_1$.

\item \textbf{Case 2:} Let $b_i$ and $c$ be such that $g_2< \frac{1}{c}$ holds:

    In this case, we have $p_1(g)>0$ and $p_1(g)^2-p_2(g)\geq 0$ for $g\in[g_2,\frac{1}{c}]$ such that the largest absolute value of the roots for $g\in[g_2,\frac{1}{c}]$ is given by 
    \begin{align}
        \bar{z}_i(g) = p_1(g)+\sqrt{p_1(g)^2-p_2(g)}.
    \end{align}
    First, note that $\bar{z}_i(g_2)=p_1(g_2)$ and next we want to compare $\bar{z}_i(g_2+\beta)$ with $\bar{z}_i(g_2)$ to see for which ${\beta\in (0,\frac{1}{c}-g_2]}$ the inequality $\bar{z}_i(g_2+\beta)\geq \bar{z}_i(g_2)$ holds.
    \begin{align}
        \bar{z}_i(g_2+\beta)&=p_1(g_2+\beta)+\sqrt{p_1(g_2+\beta)^2-p_2(g_2+\beta)} \\
        &\geq p_1(g_2+\beta)=p_1(g_2)+\frac{1+b_ic}{2}\beta\\
        &>p_1(g_2)=\bar{z}_i(g_2),
    \end{align}
    where for the first inequality we used that $\sqrt{p_1(g)^2-p_2(g)}> 0$ and for the last inequality we used that $\beta>0$.
    
    This shows us that the minimum of $\bar{z}_i(g)$ on the interval $[g_2,\frac{1}{c}]$ is achieved at $g_2$. However, from Case 1 we know that $\bar{z}_i(g)$ will further decrease on the interval $[g_1,g_2]$ with the minimum being at $g_1$ before it increases again on the interval $[0,g_1)$. 
    Hence, the minimum of $\bar{z}_i(g)$ on the interval $[0,\frac{1}{c}]$ is achieved at $g_1$.
\end{itemize}
    
In both cases, a unique local minimum exists at $g_1$, where $g_1$ is the smallest root of \eqref{eq:QuadraticEquationForOptimum}. This leads to the minimum being at sampling period ratio $N=\frac{\log(g_1)}{\log(1-\epsilon)}$. Thus, we have proved Theorem~\ref{thm:MinimumExists}.

\subsection{Proof of Theorem~\ref{thm:FiniteMinimizer}}
\label{app:ProofOfFiniteMinimizer}

Recall that from Proposition~\ref{prop:ConvergenceForLambda0} and Lemma~\ref{lem:ConvergenceOfAbsoluteValueOfRoots} we get $\lim_{N\rightarrow\infty}\bar{z}_i(N)=0$ and $\lim_{N\rightarrow\infty}\bar{z}_i(N)=|1-\lambda_i|$ for $i\neq 0$, respectively.
Thus, the objective \eqref{eq:OptimizationObjective} will converge to either $\bar{z}_1(N)$ or $\bar{z}_{n-1}(N)$ for large $N$.

First, let us assume that $|1-\lambda_1|>|1-\lambda_{n-1}|$.
Under this assumption, we recall first that $\bar{z}_1(N)$ is decreasing in $N$ (see Theorem~\ref{thm:DecreasingValue}). 
Hence, the smallest value of $\bar{z}_1(N)$ is achieved when $N\rightarrow \infty$. Furthermore, we know that \eqref{eq:OptimizationObjective} converges to $1-\lambda_1$.
So even if $\bar{z}_i(N)>\bar{z}_1(N)$ for some $N\geq h$ and $i\in \lbrace 0\rbrace \cup \mathcal{I}_{>1}$ the smallest value of the objective is achieved for $N^*=\infty$.
Hence, if $|1-\lambda_1|>|1-\lambda_{n-1}|$ holds the minimizer of the objective \eqref{eq:OptimizationObjective} is non-finite. 

Next, let us assume ${|1-\lambda_1|=|1-\lambda_{n-1}|}$.
In case ${\lambda_1\neq \lambda_{n-1}}$, we can have the same argument as in the case ${|1-\lambda_1|>|1-\lambda_{n-1}|}$ due to the decreasing nature of $\bar{z}_1(N)$, such that the optimizer $N^*$ is not finite.
In case $\lambda_1 = \lambda_{n-1}$, we know that $\lambda_1>1$ since $\lambda_{n-1}\geq \frac{n}{n-1}>1$ (see \cite{SpectralGraphTheory}). 
Hence, all modes are the same and larger than 1 and we know that a unique minimum exists for $\lambda_i\in(1,2]$ according to Theorem~\ref{thm:MinimumExists}, such that the optimizer $N^*$ is finite.

Finally, we assume $|1-\lambda_1|<|1-\lambda_{n-1}|$.
If ${\bar{z}_{n-1}(N)\geq \bar{z}_i(N)}$ for all $i$ and $N$ then we know that the minimizer $N^*$ exists based on Theorem~\ref{thm:MinimumExists}. 
If the condition is not true, we know that $\bar{z}_{n-1}(N)$ will approach its minimum as $N$ increases such that some other $\bar{z}_i(N)$ could be larger. 
However, once the minimum has been reached $\bar{z}_{n-1}(N)$ starts to grow again and converges to $|1-\lambda_{n-1}|$. 
Hence, there should be at least a point that is smaller than ${|1-\lambda_{n-1}|}$ in the function $\max_{i\in \mathcal{I}_{>1}}\left( \bar{z}_0(N), \bar{z}_1(N),\bar{z}_i(N)\right)$ for a certain $N<\infty$. 
This shows us that the minimal value of $\max_{i\in \mathcal{I}_{>1}}\left( \bar{z}_0(N), \bar{z}_1(N),\bar{z}_i(N)\right)$ exists for a finite $N$.



\bibliographystyle{IEEEtran}
\bibliography{biblio}

\end{document}